\documentclass{llncs}
\usepackage{version}
\usepackage{pais}
\usepackage{url}
\title{Contradiction-Tolerant Process Algebra \\ with Propositional
       Signals}
\author{J.A. Bergstra \and C.A. Middelburg}
\institute{Informatics Institute, Faculty of Science, University of
           Amsterdam, \\
           Science Park~904, 1098~XH Amsterdam, the Netherlands \\
           \email{J.A.Bergstra@uva.nl,C.A.Middelburg@uva.nl}}

\begin{document}
\maketitle

\begin{abstract}
In a previous paper, an ACP-style process algebra was proposed in which 
propositions are used as the visible part of the state of processes and 
as state conditions under which processes may proceed.
This process algebra, called ACPps, is built on classical propositional
logic.
In this paper, we present a version of ACPps built on a paraconsistent
propositional logic which is essentially the same as CLuNs.
There are many systems that would have to deal with self-contradictory 
states if no special measures were taken.
For a number of these systems, it is conceivable that accepting 
self-contradictory states and dealing with them in a way based on a 
paraconsistent logic is an alternative to taking special measures. 
The presented version of ACPps can be suited for the description and 
analysis of systems that deal with self-contradictory states in a 
way based on the above-mentioned paraconsistent logic.
\begin{classcode}
D.2.1, D.2.4, F.3.1, F.4.1.
\end{classcode}
\end{abstract}

\section{Introduction}
\label{sect-intro}

Algebraic theories of processes such as ACP~\cite{BW90}, 
CCS~\cite{Mil89}, and CSP~\cite{Hoa85}, as well as most algebraic 
theories of processes in the style of these ones, are concerned with 
the behaviour of processes only.
That is, the state of processes is kept invisible.
In~\cite{BB94b}, an ACP-style process algebra, called ACPps, was 
proposed in which processes have their state to some extent visible.
The visible part of the state of a process, called the signal emitted by 
the process, is a proposition of classical propositional logic.
Propositions are not only used as signals emitted by processes, but also
as conditions under which processes may proceed.
The intuition is that the signal emitted by a process is a proposition 
that holds at its start and the condition under which processes may 
proceed is a proposition that must hold at its start.
Thus, by the introduction of signal emitting processes, an answer is 
given to the question what determines whether a condition under which a 
process may proceed is met.

If the signals emitted by two processes are contradictory, then the 
signal emitted by the parallel composition of these processes is 
self-contradictory.
For example, if the signals emitted by the two processes, being 
propositions, are each others negation, then they are contradictory 
and their conjunction, which is the signal emitted by the parallel 
composition of these processes, is self-contradictory.
Intuitively, a process emitting a self-contradictory signal is an 
impossibility.
Therefore, a special process has been introduced in ACPps to deal with 
it.
In practice, there are many systems that would have to deal with  
self-contradictory states if no special prevention measures or special 
detection and resolution measures were taken.
Some typical examples are web-service-oriented applications and 
autonomous robotic agents (see e.g.~\cite{GKNF05a,PI04a,QV14a}).
At least for a number of these systems, it is conceivable that accepting 
self-contradictory states and dealing with them in a way based on a 
suitable paraconsistent logic is an alternative to taking special
measures. 
It may even be the only workable alternative because a system may have 
to cope with inconsistencies occurring on a large scale.

What exactly does it mean to deal with self-contradictory states in a 
way based on a paraconsistent logic?
The systems referred to above are systems whose behaviour is made up of 
discrete steps where, upon each step performed, the way in which the 
behaviour proceeds is conditional on the current state of the system
concerned.
If the propositions by which the visible part of the possible states of 
a system can be characterized are used as conditions, then it can be 
established in accordance with a paraconsistent propositional logic 
whether a condition is met in a state.
This is what is meant by dealing with self-contradictory states in a way 
based on a paraconsistent logic. 
We think that a version of the process algebra ACPps that is built on an 
appropriate paraconsistent propositional logic instead of classical
propositional logic can be suited for the description and analysis of 
systems that deal with self-contradictory states in a way based on a 
paraconsistent logic.
The important point here is that, in such a logic, it is generally not 
possible to deduce an arbitrary formula from two contradictory formulas.

The question remains: what is an appropriate paraconsistent 
propositional logic?
The ones that have been proposed differ in many ways and whether one of 
them is more appropriate than another is fairly difficult to make out.
A paraconsistent propositional logic is a logic that does not have the 
property that every proposition is a logical consequence of every set of 
hypotheses that contains contradictory propositions.
A paraconsistent propositional logic with the property that every 
proposition is a logical consequence of every set of hypotheses that 
contains contradictory propositions but one is far from appropriate.
Such a logic is a minimal paraconsistent logic.
Maximal paraconsistency, i.e.\ a logical consequence relation that 
cannot be extended without loosing paraconsistency, is generally 
considered an important property.
There are various other properties that have been proposed as 
characteristic of reasonable paraconsistent propositional logics, but 
their importance remains to some extent open to question.
 
The properties that have been proposed as characteristic of reasonable 
paraconsistent propositional logics do not include all properties that 
are required of an appropriate one to build a version of ACPps on.
These properties include, among other things, properties needed to 
retain the basic axioms of ACP-style process algebras.
In this paper, we present a version of ACPps built on the paraconsistent
propositional logic for which the name LP$^{\IImpl,\False}$ was coined
in~\cite{Mid11a}.
This logic, which is essentially the same as J3~\cite{DOt85a}, 
CLuNs~\cite{BC04a}, and LFI1~\cite{CCM07a}, has virtually all properties 
that have been proposed as characteristic of reasonable paraconsistent 
propositional logics as well as all properties that are required of an 
appropriate one to build a version of ACPps on.
LP$^{\IImpl,\False}$ can be replaced by any paraconsistent propositional 
logics with the latter properties, but among the paraconsistent 
propositional logics with the former properties, LP$^{\IImpl,\False}$ is 
the only one with the latter properties.

The structure of this paper is as follows.
First, we give a survey of the paraconsistent propositional logic 
LP$^{\IImpl,\False}$ (Section~\ref{sect-LP-iimpl-false}).
Next, we present \ctBPAps, the subtheory of the version of ACPps built 
on LP$^{\IImpl,\False}$ that does not support parallelism and 
communication (Sections~\ref{sect-ctBPAps} and~\ref{sect-sem-ctBPAps}).
After that, we present \ctACPps, the version of ACPps built on 
LP$^{\IImpl,\False}$, as an extension of \ctBPAps\
(Sections~\ref{sect-ctACPps} and~\ref{sect-sem-ctACPps}).
Following this, we introduce a useful additional feature, namely a 
generalization of the state operators from~\cite{BB88} 
(Section~\ref{sect-ctACPps+SO}).
Then, we treat the addition of guarded recursion to \ctACPps\ 
(Section~\ref{sect-ctACPps+REC}).
Finally, we make some concluding remarks (Section~\ref{sect-concl}).

\section{The Paraconsistent Logic LP$^{\IImpl,\False}$}
\label{sect-LP-iimpl-false}

A set of propositions $\Gamma$ is contradictory if there exists a 
proposition $A$ such that both $A$ and $\Not A$ can be deduced from 
$\Gamma$.
A proposition $A$ is called self-contradictory if $\set{A}$ is 
contradictory.
In classical propositional logic, every proposition can be deduced from 
a contradictory set of propositions.
A paraconsistent propositional logic is a propositional logic in which 
not every proposition can be deduced from each contradictory set of 
propositions.

In~\cite{Pri79a}, Priest proposed the paraconsistent propositional logic
LP (Logic of Paradox).
The logic introduced in this section is LP enriched with an implication 
connective for which the standard deduction theorem holds and a falsity 
constant. 
This logic, called LP$^{\IImpl,\False}$, is in fact the propositional 
fragment of CLuNs~\cite{BC04a} without bi-implications.

LP$^{\IImpl,\False}$ has the following logical constants and connectives:
a falsity constant $\False$,
a unary negation connective $\Not$, 
a binary conjunction connective $\CAnd$, 
a binary disjunction connective $\COr$, and
a binary implication connective $\IImpl$.
Truth and bi-implication are defined as abbreviations:
$\True$ stands for $\Not \False$ and
$A \BIImpl B$ stands for $(A \IImpl B) \CAnd (B \IImpl A)$.

A Hilbert-style formulation of LP$^{\IImpl,\False}$ is given in
Table~\ref{proofsystem-LPiimpl}.
\begin{table}[!tb]
\caption{Hilbert-style formulation of LP$^{\IImpl,\False}$}
\label{proofsystem-LPiimpl}
\begin{eqntbl}
\begin{axcol}
\mathbf{Axiom\; Schemas:}
\\
A \IImpl (B \IImpl A)
\\
(A \IImpl (B \IImpl C)) \IImpl ((A \IImpl B) \IImpl (A \IImpl C))
\\
((A \IImpl B) \IImpl A) \IImpl A
\\
\False \IImpl A
\\
(A \CAnd B) \IImpl A
\\
(A \CAnd B) \IImpl B
\\
A \IImpl (B \IImpl (A \CAnd B))
\\
A \IImpl (A \COr B)
\\
B \IImpl (A \COr B)
\\
(A \IImpl C) \IImpl ((B \IImpl C) \IImpl ((A \COr B) \IImpl C))
\end{axcol}
\qquad
\begin{axcol}
{}
\\
\Not \Not A \BIImpl A
\\
\Not (A \IImpl B) \BIImpl A \CAnd \Not B
\\
\Not (A \CAnd B) \BIImpl \Not A \COr \Not B
\\
\Not (A \COr B) \BIImpl \Not A \CAnd \Not B
\\
{}
\\
A \COr \Not A
\\
{}
\\
\mathbf{Rule\; of\; Inference:}
\\
\Infrule{A \quad A \IImpl B}{B}
\end{axcol}
\end{eqntbl}
\end{table}
\sloppy
In this formulation, which is taken from~\cite{Avr91a}, $A$, $B$, and 
$C$ are used as meta-variables ranging over all formulas of
LP$^{\IImpl,\False}$.
The axiom schemas on the left-hand side of
Table~\ref{proofsystem-LPiimpl} and the single inference rule (modus
ponens) constitute a Hilbert-style formulation of the positive fragment
of classical propositional logic.
The first four axiom schemas on the right-hand side of
Table~\ref{proofsystem-LPiimpl} allow for the negation connective to be
moved inward.
The fifth axiom schema on the right-hand side of
Table~\ref{proofsystem-LPiimpl} is the law of the excluded middle.
This axiom schema can be thought of as saying that, for every
proposition, the proposition or its negation is true, while leaving open
the possibility that both are true.
If we add the axiom schema $\Not A \IImpl (A \IImpl B)$, which says that
any proposition follows from a contradiction, to the given Hilbert-style
formulation of LP$^{\IImpl,\False}$, then we get a Hilbert-style 
formulation of classical propositional logic (see e.g.~\cite{Avr91a}).
We write $\pEnt$ for the syntactic logical consequence relation induced 
by the axiom schemas and inference rule of LP$^{\IImpl,\False}$.

The following outline of the semantics of LP$^{\IImpl,\False}$ is based
on~\cite{Avr91a}.
Like in the case of classical propositional logic, meanings are assigned
to the formulas of LP$^{\IImpl,\False}$ by means of valuations.
However, in addition to the two classical truth values $\true$ (true)
and $\false$ (false), a third meaning $\both$ (both true and false) may
be assigned.
A \emph{valuation} for LP$^{\IImpl,\False}$ is a function $\nu$ from the 
set of all formulas of LP$^{\IImpl,\False}$ to the set 
$\{\true,\false,\both\}$ such that for all formulas $A$ and $B$ of 
LP$^{\IImpl,\False}$:
\begin{eqnarray*}
\val{\False}{\nu} & = & \false,
\\
\val{\Not A}{\nu} & = &
 \left \{
 \begin{array}{l@{\;\;}l}
 \true  & \mathrm{if}\; \val{A}{\nu} = \false \\
 \false & \mathrm{if}\; \val{A}{\nu} = \true \\
 \both  & \mathrm{otherwise},
 \end{array}
 \right.
\\
\val{A \CAnd B}{\nu} & = &
 \left \{
 \begin{array}{l@{\;\;}l}
 \true  & \mathrm{if}\; \val{A}{\nu} = \true  \;\mathrm{and}\;
                        \val{B}{\nu} = \true  \\
 \false & \mathrm{if}\; \val{A}{\nu} = \false \;\mathrm{or}\;
                        \val{B}{\nu} = \false \\
 \both  & \mathrm{otherwise},
 \end{array}
 \right.
\\
\val{A \COr B}{\nu} & = &
 \left \{
 \begin{array}{l@{\;\;}l}
 \true  & \mathrm{if}\; \val{A}{\nu} = \true  \;\mathrm{or}\;
                        \val{B}{\nu} = \true  \\
 \false & \mathrm{if}\; \val{A}{\nu} = \false \;\mathrm{and}\;
                        \val{B}{\nu} = \false \\
 \both  & \mathrm{otherwise},
 \end{array}
 \right.
\\
\val{A \IImpl B}{\nu} & = &
 \left \{
 \begin{array}{l@{\;\;}l}
 \true           & \mathrm{if}\; \val{A}{\nu} = \false \\
 \val{B}{\nu} & \mathrm{otherwise}.
 \end{array}
 \right.
\end{eqnarray*}
The classical truth-conditions and falsehood-conditions for the logical
connectives are retained.
Except for implications, a formula is classified as both-true-and-false
exactly when when it cannot be classified as true or false by the
classical truth-conditions and falsehood-conditions.
The definition of a valuation given above shows that the logical 
connectives of LP$^{\IImpl,\False}$ are (three-valued) truth-functional, 
which means that each $n$-ary connective represents a function from 
$\{\true,\false,\both\}^n$ to $\{\true,\false,\both\}$.

For LP$^{\IImpl,\False}$, the semantic logical consequence relation, 
denoted by $\mEnt$, is based on the idea that a valuation $\nu$ 
satisfies a formula $A$ if $\val{A}{\nu} \in \{\true,\both\}$.
It is defined as follows: $\Gamma \mEnt A$ iff for every 
valuation $\nu$, either $\val{A'}{\nu} = \false$ for some 
$A' \in \Gamma$ or $\val{A}{\nu} \in \{\true,\both\}$.
We have that the Hilbert-style formulation of LP$^{\IImpl,\False}$ is 
strongly complete with respect to its semantics, i.e.\ $\Gamma \pEnt A$ 
iff $\Gamma \mEnt A$ (see e.g.~\cite{BC04a}).

For all formulas $A$ of LP$^{\IImpl,\False}$ in which $\False$ does not 
occur, for all formulas $B$ of LP$^{\IImpl,\False}$ in which no 
propositional variable occurs that occurs in $A$, 
$\set{A, \Not A} \npEnt B$ if $\npEnt B$ (see e.g.~\cite{AA15a}).%
\footnote
{We use the notation 
 ${} \pEnt A$ for $\emptyset\hspace*{.02em} \pEnt A$,
 ${} \npEnt A$ for not $\emptyset\hspace*{.02em} \pEnt A$, and
 $\Gamma \npEnt A$ for not $\Gamma \pEnt A$.}
Hence, LP$^{\IImpl,\False}$ is a paraconsistent propositional logic.

For LP$^{\IImpl,\False}$, the logical equivalence relation $\LEqv$ is 
defined as for classical propositional logic: $A \LEqv B$ iff for every 
valuation $\nu$, $\val{A}{\nu} = \val{B}{\nu}$.
Unlike in classical propositional logic, we do not have that $A \LEqv B$
iff ${} \pEnt A \BIImpl B$.

For LP$^{\IImpl,\False}$, the consistency property is defined as to be 
expected: $A$ is consistent iff for every valuation $\nu$, 
$\val{A}{\nu} \neq \both$.

The following are some important properties of LP$^{\IImpl,\False}$:
\begin{list}{}
 {\setlength{\leftmargin}{2.25em} \settowidth{\labelwidth}{(g)}}
\item[(a)]
\emph{containment in classical logic}:
${\pEnt} \subseteq {\clpEnt}$;%
\footnote
{We use the symbol $\clpEnt$ to denote the logical consequence relation 
 of classical propositional logic.}
\item[(b)]
\emph{proper basic connectives}:
for all sets $\Gamma$ of formulas of LP$^{\IImpl,\False}$ and
all formulas $A$, $B$, and $C$ of LP$^{\IImpl,\False}$:
\begin{list}{}
 {\setlength{\leftmargin}{2.5em} \settowidth{\labelwidth}{(c)}}
\item[(b$_1$)]
$\Gamma \union \set{A} \pEnt B$\phantom{${} \CAnd {}$}\phantom{$C$} iff 
$\Gamma \pEnt A \IImpl B$,
\item[(b$_2$)]
$\Gamma \pEnt A \CAnd B$\phantom{,}\phantom{$C$} iff 
$\Gamma \pEnt A$ and $\Gamma \pEnt B$,
\item[(b$_3$)]
$\Gamma \union \set{A \COr B} \pEnt C$ iff 
$\Gamma \union \set{A} \pEnt C$ and $\Gamma \union \set{B} \pEnt C$;
\end{list}
\item[(c)]
\emph{weakly maximal paraconsistency relative to classical logic}:
for all formulas $A$ of LP$^{\IImpl,\False}$ with $\npEnt A$ and 
$\clpEnt A$, for the minimal consequence relation $\extpEnt$ such that
${\pEnt} \subseteq {\extpEnt}$ and $\extpEnt A$, for all formulas $B$ 
of LP$^{\IImpl,\False}$, $\extpEnt B$ iff $\clpEnt B$;
\item[(d)]
\emph{strongly maximal absolute paraconsistency}:
for all logics $\mathcal{L}$ with the same logical constants and 
connectives as LP$^{\IImpl,\False}$ and a consequence relation 
$\extpEnt$ such that ${\pEnt} \subset {\extpEnt}$, $\mathcal{L}$ is not 
paraconsistent;
\item[(e)]
\emph{internalized notion of consistency}: $A$ is consistent iff
${} \pEnt (A \IImpl \False) \COr (\Not A \IImpl \False)$;
\item[(f)]
\emph{internalized notion of logical equivalence}: $A \LEqv B$ iff
${} \pEnt (A \BIImpl B) \CAnd (\Not A \BIImpl \Not B)$;
\item[(g)]
the laws given in Table~\ref{laws-lequiv} hold for the logical 
equivalence relation of LP$^{\IImpl,\False}$.%
\begin{table}[!tb]
\caption{Laws that hold for the logical equivalence relation of
 LP$^{\IImpl,\False}$}
\label{laws-lequiv}
\begin{eqntbl}
\begin{neqncol}
(1)  & A \CAnd \False \LEqv \False \\
(3)  & A \CAnd \True \LEqv A \\
(5)  & A \CAnd A \LEqv A \\
(7)  & A \CAnd B \LEqv B \CAnd A \\
(9)  & (A \CAnd B) \CAnd C \LEqv A \CAnd (B \CAnd C) \\
(11) & A \CAnd (B \COr C) \LEqv (A \CAnd B) \COr (A \CAnd C) \\
(13) & (A \IImpl B) \CAnd (A \IImpl C) \LEqv A \IImpl (B \CAnd C) \\
(15) & (A \COr \Not A) \IImpl B \LEqv B
\end{neqncol}
\qquad
\begin{neqncol}
(2)  & A \COr \True \LEqv \True \\
(4)  & A \COr \False \LEqv A \\
(6)  & A \COr A \LEqv A \\
(8)  & A \COr B \LEqv B \COr A \\
(10) & (A \COr B) \COr C \LEqv A \COr (B \COr C) \\
(12) & A \COr (B \CAnd C) \LEqv (A \COr B) \CAnd (A \COr C) \\
(14) & (A \IImpl C) \CAnd (B \IImpl C) \LEqv (A \COr B) \IImpl C \\
(16) & A \IImpl (B \IImpl C) \LEqv (A \CAnd B) \IImpl C 
\end{neqncol}
\end{eqntbl}
\end{table}
\end{list}
Properties~(a)--(f) have been mentioned relatively often in the 
literature (see e.g.~\cite{AA15a,AAZ11b,AAZ11a,Avr99a,BC04a,CCM07a}).
Properties~(a), (b$_1$), (c), and~(d) make LP$^{\IImpl,\False}$ an ideal 
paraconsistent logic in the sense made precise in~\cite{AAZ11b}.
By property~(e), LP$^{\IImpl,\False}$ is also a logic of formal 
inconsistency in the sense made precise in~\cite{CCM07a}.
Properties~(a)--(c) \linebreak[2] indicate that LP$^{\IImpl,\False}$ 
retains much of classical propositional logic.
Actually, property~(c) can be strengthened to the following property: 
for all formulas $A$ of LP$^{\IImpl,\False}$, $\pEnt A$ iff $\clpEnt A$.

From Theorem~4.42 in~\cite{AA15a}, we know that there are exactly 8192 
different three-valued paraconsistent propositional logics with 
properties~(a) and~(b).
From Theorem~2 in~\cite{AAZ11b}, we know that properties~(c) and~(d) are 
common properties of all three-valued paraconsistent propositional 
logics with properties~(a) and~(b$_1$).
From Fact~103 in~\cite{CCM07a}, we know that property~(f) is a common 
property of all three-valued paraconsistent propositional logics with 
properties~(a), (b) and~(e).
Moreover, it is easy to see that that property~(e) is a common property 
of all three-valued paraconsistent propositional logics with 
properties~(a) and~(b).
Hence, each three-valued paraconsistent propositional logic with 
properties~(a) and~(b) has properties~(c)--(f) as well.

Property~(g) is not a common property of all three-valued paraconsistent 
propositional logics with properties~(a) and~(b).
To our knowledge, properties like property~(g) are not mentioned in the 
literature.
However, like property~(f), property~(g) is essential for the process 
algebra presented in this paper.
Among the 8192 three-valued paraconsistent propositional logics with 
properties~(a)--(e), \linebreak[2] which are considered desirable 
properties, LP$^{\IImpl,\False}$ is one out of four with the essential 
properties~(f) and~(g).
\begin{proposition}[Almost Uniqueness]
\label{proposition-uniqueness}
There are exactly four three-valued paraconsistent propositional logics 
with the logical constants and connectives of LP$^{\IImpl,\False}$ that 
have the properties~(a)--(g) mentioned above.
\end{proposition}
\begin{proof}
Because property~(f) is a common property of all 8192 three-valued
paraconsistent propositional logics with properties~(a)--(e), it is 
sufficient to prove that, among these 8192 logics, there exists only 
one that has property~(g).
Because `non-deterministic truth tables' that uniquely characterize the 
8192 logics are given in~\cite{AAZ11b}, the theorem can be proved by 
showing that, for each of the connectives, only one of the ordinary 
truth tables represented by the non-deterministic truth table for that 
connective is compatible with the laws given in Table~\ref{laws-lequiv}.
It can be shown by short routine case analyses that only one of the 8
ordinary truth tables represented by the non-deterministic truth tables 
for conjunction is compatible with laws~(1), (3), (5), and~(7) and only 
one of the 32 ordinary truth tables represented by the non-deterministic 
truth tables for disjunction is compatible with laws~(2), (4), (6), 
and~(8).
The truth tables concerned are compatible with laws~(9)--(12) as well. 
Given the ordinary truth table for conjunction and disjunction so 
obtained, it can be shown by slightly longer routine case analyses that 
exactly four of the 16 ordinary truth tables represented by the 
non-deterministic truth table for implication are compatible with 
laws~(13)--(15).
The four truth tables concerned are compatible with law~(16) as well.
\qed
\end{proof}
The next corollary follows from the proof of 
Proposition~\ref{proposition-uniqueness}.
\begin{corollary}[Uniqueness]
\label{corollary-uniqueness}
LP$^{\IImpl,\False}$ is the only three-valued paraconsistent 
pro\-positional logic with the logical constants and connectives of 
LP$^{\IImpl,\False}$ that has the properties~(a)--(g) mentioned above 
and moreover the property that the law $\Not \Not A \LEqv A$ holds for 
its logical equivalence relation.
\end{corollary}
Corollary~\ref{corollary-uniqueness} may be of independent importance to 
the area of paraconsistent logics.

From now on, we will use the following abbreviations: $A \Iff B$ stands 
for $(A \BIImpl B) \CAnd (\Not A \BIImpl \Not B)$ and $\Cons A$ stands 
for $(A \IImpl \False) \COr (\Not A \IImpl \False)$.

In Section~\ref{sect-ctBPAps}, where we will use formulas of 
LP$^{\IImpl,\False}$ as terms, equality of formulas will be interpreted 
as logical equivalence.
This means that equality of formulas can be formally proved using the 
fact that $A \LEqv B$ iff $\pEnt A \Iff B$.
This fact also suggests that LP$^{\IImpl,\False}$ may be Blok-Pigozzi 
algebraizable~\cite{BP89a}.
It is shown in~\cite{CCM07a} that actually all 8192 three-valued 
paraconsistent propositional logics referred to above are Blok-Pigozzi 
algebraizable.
Although there must exist one, a conditional-equational axiomatization 
of the algebras concerned in the case of LP$^{\IImpl,\False}$ has not 
yet been devised.
Owing to this, the equations derivable in the version of ACPps built on 
LP$^{\IImpl,\False}$ presented in this paper cannot always be derived by 
equational reasoning only. 

\section{Contradiction-Tolerant BPA with Propositional Signals}
\label{sect-ctBPAps}

BPAps is a subtheory of ACPps that does not support parallelism and
communication.
In this section, we present the contradiction-tolerant version of BPAps.
In this version, which is called \ctBPAps, processes have their state to 
some extent visible.
The visible part of the state of a process, called the signal emitted by 
the process, is a proposition of LP$^{\IImpl,\False}$.
These propositions are not only used as signals emitted by processes, 
but also as conditions under which processes may proceed.
The intuition is that the signal emitted by a process is a proposition 
that holds at its start and the condition under which processes may 
proceed is a proposition that must holds at its start.

In \ctBPAps, just as in BPAps, it is assumed that a fixed but arbitrary 
finite set $\Act$ of \emph{actions}, with $\dead \not\in \Act$, and a 
fixed but arbitrary finite set $\AProp$ of \emph{atomic propositions} 
have been given.
We write $\Actd$ for $\Act \union \set{\dead}$.

The algebraic theory \ctBPAps\ has two sorts:
\begin{itemize}
\item
the sort $\Proc$ of \emph{processes};
\item
the sort $\Prop$ of \emph{propositions}.
\end{itemize}
\pagebreak[2]
The algebraic theory \ctBPAps\ has the following constants and 
operators to build terms of sort $\Prop$:
\begin{itemize}
\item
for each $P \in \AProp$, the \emph{atomic proposition} constant 
$\const{P}{\Prop}$;
\item
the \emph{falsity} constant $\const{\False}{\Prop}$;
\item
the unary \emph{negation} operator $\funct{\Not}{\Prop}{\Prop}$;
\item
the binary \emph{conjunction} operator 
$\funct{\CAnd}{\Prop \x \Prop}{\Prop}$;
\item
the binary \emph{disjunction} operator 
$\funct{\COr}{\Prop \x \Prop}{\Prop}$;
\item
the binary \emph{implication} operator 
$\funct{\IImpl}{\Prop \x \Prop}{\Prop}$.
\end{itemize}
The algebraic theory \ctBPAps\ has the following constants and 
operators to  build terms of sort $\Proc$:
\begin{itemize}
\item
the \emph{deadlock} constant $\const{\dead}{\Proc}$;
\item
for each $a \in \Act$, the \emph{action} constant $\const{a}{\Proc}$;
\item
the \emph{inaccessible process} constant $\const{\nex}{\Proc}$;
\item
the binary \emph{alternative composition} operator 
$\funct{\altc}{\Proc \x \Proc}{\Proc}$;
\item
the binary \emph{sequential composition} operator 
$\funct{\seqc}{\Proc \x \Proc}{\Proc}$;
\item
the binary \emph{guarded command} operator 
$\funct{\gc}{\Prop \x \Proc}{\Proc}$;
\item
the binary \emph{signal emission} operator 
$\funct{\emi}{\Prop \x \Proc}{\Proc}$.
\end{itemize}
It is assumed that there are infinitely many variables of sort $\Proc$, 
including $x$, $y$, and $z$. 

We use infix notation for the binary operators.
The following precedence conventions are used to reduce the need for
parentheses.
The operators to build terms of sort $\Prop$ bind stronger than the
operators to build terms of sort $\Proc$.
The operator ${} \seqc {}$ binds stronger than all other binary 
operators to build terms of sort $\Proc$ and the operator ${} \altc {}$
binds weaker than all other binary operators to build terms of sort 
$\Proc$.

Let $p$ and $q$ be closed terms of sort $\Proc$ and $\phi$ be a 
closed term of sort $\Prop$.
Intuitively, the constants and operators to build terms of sort $\Proc$
can be explained as follows:
\begin{itemize}
\item
$\dead$ is not capable of doing anything, the proposition that holds at 
the start of $\dead$ is $\True$;
\item
$a$ is only capable of performing action $a$ unconditionally and next 
terminating successfully, the proposition that holds at the start of $a$ 
is $\True$;
\item
$\nex$ is not capable of doing anything; there is an inconsistency at 
the start of~$\nex$;
\item
$p \altc q$ behaves either as $p$ or as $q$ but not both, the 
proposition that holds at the start of $p \altc q$ is the conjunction of 
the propositions that hold at the start of $p$ and $q$;
\item
$p \seqc q$ first behaves as $p$ and on successful termination of $p$ 
it next behaves as $q$, the proposition that holds at the start of 
$p \seqc q$ is the proposition that holds at the start of $p$;
\item
$\phi \gc p$ behaves as $p$ under condition $\phi$, the proposition 
that holds at the start of $\phi \gc p$ is the implication with $\phi$ 
as antecedent and the proposition that holds at the start of $p$ as 
consequent;
\item
$\phi \emi p$ behaves as $p$ if the proposition that holds at its start 
does not equal $\False$ and as $\nex$ otherwise, in the former case, the 
proposition that holds at the start of $\phi \emi p$ is the conjunction 
of $\phi$ and the proposition that holds at the start of $p$.
\end{itemize}

The axioms of \ctBPAps\ are the axioms given in 
Table~\ref{axioms-ctBPAps}.%
\begin{table}[!tb]
\caption{Axioms of \ctBPAps}
\label{axioms-ctBPAps}
\begin{eqntbl}
\begin{axcol}
x \altc y = y \altc x                                 & \ax{A1}\\
(x \altc y) \altc z = x \altc (y \altc z)             & \ax{A2}\\
x \altc x = x                                         & \ax{A3}\\
(x \altc y) \seqc z = x \seqc z \altc y \seqc z       & \ax{A4}\\
(x \seqc y) \seqc z = x \seqc (y \seqc z)             & \ax{A5}\\
x \altc \dead = x                                     & \ax{A6}\\
\dead \seqc x = \dead                                 & \ax{A7}\\
{}\\
\True \gc x = x                                       & \ax{GC1}\\
\False \gc x = \dead                                  & \ax{GC2}\\
\phi \gc \dead = \dead                                & \ax{GC3}\\
\phi \gc (x \altc y) = \phi \gc x \altc \phi \gc y    & \ax{GC4}\\
\phi \gc x \seqc y = (\phi \gc x) \seqc y             & \ax{GC5}\\
\phi \gc (\psi \gc x) = (\phi \CAnd \psi) \gc x        & \ax{GC6}\\
(\phi \COr \psi) \gc x = \phi \gc x \altc \psi \gc x   & \ax{GC7}
\end{axcol}
\qquad
\begin{axcol}
x \altc \nex = \nex                                   & \ax{NE1}\\
\nex \seqc x = \nex                                   & \ax{NE2}\\
a \seqc \nex = \dead                                  & \ax{NE3}\\
{}\\{}\\{} \\
\phi = \psi \hfill 
\mif \pEnt \phi \Iff \psi & \ax{IMP}\\
{}\\
\True \emi x = x                                      & \ax{SE1}\\
\False \emi x = \nex                                  & \ax{SE2}\\
\phi \emi \nex = \nex                                 & \ax{SE3}\\
\phi \emi x \altc y = \phi \emi (x \altc y)           & \ax{SE4}\\
(\phi \emi x) \seqc y = \phi \emi x \seqc y           & \ax{SE5}\\
\phi \emi (\psi \emi x ) = (\phi \CAnd \psi) \emi x    & \ax{SE6}\\
\phi \emi (\phi \gc x) = \phi \emi  x                 & \ax{SE7}\\
\phi \gc (\psi \emi x) = 
(\phi \IImpl \psi) \emi (\phi \gc x)                  & \ax{SE8}
\end{axcol}
\end{eqntbl}
\end{table}
In this table, $a$ stands for an arbitrary constant from 
$\Act \union \set{\dead}$, $\phi$ and $\psi$ stand for arbitrary closed 
terms of sort $\Prop$, and ${} \pEnt {}$ is 
the logical consequence relation of LP$^{\IImpl,\False}$.
A1--A7 are the axioms of \BPAd, the subtheory of ACP that does not 
support parallelism and communication (see e.g.~\cite{BW90}).
NE1--NE3, GC1--GC7, and SE1--SE8 have been taken from~\cite{BB94b}, 
using a different numbering.%
\footnote
{The axioms of \ctBPAps\ are not independent:
 A3, A6, and A7 are derivable from GC1--GC7 and IMP,
 NE1 and NE2 are derivable from SE1--SE8, and
 SE3 is derivable from SE6 and IMP.}
By IMP, the axioms of \ctBPAps\ include all equations $\phi = \psi$ for 
which $\phi \Iff \psi$ is a theorem of LP$^{\IImpl,\False}$.
This is harmless because the connective $\Iff$, which is the 
internalization of the logical equivalence relation $\LEqv$ of 
LP$^{\IImpl,\False}$, is a congruence.

The following generalizations of axioms SE4 and SE7 are among the 
equations derivable from the axioms of \ctBPAps:
\begin{ldispl}
\begin{geqns}
\phi \emi x \altc \psi \emi y = 
  (\phi \CAnd \psi) \emi (x \altc y)\;,
\\
(\phi \CAnd \psi) \emi (\phi \gc x) = (\phi \CAnd \psi) \emi x\;,
\\
\phi \emi ((\phi \CAnd \psi) \gc x) = \phi \emi (\psi \gc x)\;;
\end{geqns}
\end{ldispl}%
the following specialization of axiom SE4 is among the equations 
derivable from the axioms of \ctBPAps:
\begin{ldispl}
\begin{geqns}
\phi \emi \dead \altc x = \phi \emi x\;;
\end{geqns}
\end{ldispl}%
and the following equations concerning the inaccessible process are 
among the equations derivable from the axioms of \ctBPAps: 
\begin{ldispl}
\begin{geqns}
\phi \emi \nex = \nex\;,
\\
\phi \gc \nex = (\phi \IImpl \False) \emi \dead\;.
\end{geqns}
\end{ldispl}%
The derivable equations mentioned above are derivable from the axioms of 
BPAps as well.
The equation $\phi \gc \nex = \Not \phi \emi \dead$, which is derivable 
from the axioms of BPAps, is not derivable from the axioms of \ctBPAps.

Let $\phi$ be a closed term of sort $\Prop$ such that
not $\pEnt \phi \,\Iff\, \False$ and
not $\pEnt \Not \phi \,\Iff\, \False$. 
Then, because not
$\pEnt \phi \CAnd \Not \phi \,\Iff\, \False$, 
we have that
$a \seqc (\phi \emi x \altc \Not \phi \emi y) = 
 a \seqc (\False \emi (x \altc y)) =  \dead$,
which is derivable from the axioms of BPAps, is not derivable from the 
axioms of \ctBPAps. 
This shows the main difference between \ctBPAps\ and BPAps: the 
alternative composition of two processes of which the propositions that 
hold at the start of them are contradictory does not lead to an 
inconsistency in \ctBPAps, whereas it does lead to an inconsistency in 
BPAps.
This is why \ctBPAps\ is called the contradiction-tolerant version of
BPAps.

Let $\phi$ be a closed term of sort $\Prop$ such that
not $\pEnt \phi \,\Iff\, \False$ and
not $\pEnt \Not \phi \,\Iff\, \False$. 
We can derive 
$a \seqc (\phi \emi b \altc \Not \phi \emi c) =
 a \seqc ((\phi \CAnd \Not \phi) \emi (b \altc c)) = \dead$ 
from the axioms of BPAps because, in the case of BPAps, 
$a \seqc (\phi \emi b \altc \Not \phi \emi c)$ 
is not capable of doing anything.
We can only derive
$a \seqc (\phi \emi b \altc \Not \phi \emi c) =
 a \seqc ((\phi \CAnd \Not \phi) \emi (b \altc c))$ 
from the axioms of \ctBPAps\ because, in the case of \ctBPAps, 
$a \seqc (\phi \emi b \altc \Not \phi \emi c)$ 
is capable of first performing $a$ and next either performing $b$ and 
after that terminating successfully or performing $c$ and after that 
terminating successfully --- although the proposition that holds at 
the start of the process that remains after performing $a$ is the 
contradiction $\phi \CAnd \Not \phi$.

Let $\phi$ be a closed term of sort $\Prop$ such that
not $\pEnt \phi \,\Iff\, \False$ and
not $\pEnt \Not \phi \,\Iff\, \False$. 
Then, because
$\pEnt \Cons \phi \CAnd \phi \CAnd \Not \phi \,\Iff\, \False$,  
we have that 
$a \seqc (\Cons \phi \emi (\phi \emi x \altc \Not \phi \emi y)) =
 a \seqc (\False \emi (x \altc y)) = \dead$
is derivable from the axioms of \ctBPAps. 
This shows that it can be enforced by means of a consistency proposition
($\Cons \phi$) that the alternative composition of two processes of 
which the propositions that hold at the start of them are contradictory 
leads to an inconsistency in \ctBPAps.

Hereafter, we will write $[\phi]$ for the equivalence class of $\phi$ 
modulo $\LEqv$.
That is, $[\phi] = \set{\psi \where \phi \LEqv \psi}$.
Hence, 
$[\phi] =
 \set{\psi \where \;\pEnt \phi \Iff \psi}$.

All processes that can be described by a closed term of \ctBPAps, can be
described by a basic term.
The set $\cB$ of \emph{basic terms} is inductively defined by the 
following rules:
\begin{itemize}
\item
$\nex \in \cB$;
\item
if $\phi \notin [\False]$, then $\phi \emi \dead \in \cB$;
\item
if $\phi \notin [\False]$ and $a \in \Act$, then $\phi \gc a \in \cB$;
\item
if $\phi \notin [\False]$, $a \in \Act$, and $p \in \cB$, then  
$\phi \gc a \seqc p \in \cB$;
\item
if $p,q \in \cB$, then $p \altc q \in \cB$.
\end{itemize}
\pagebreak[2]
Each basic term can be written as $\nex$ or in the form
\begin{ldispl}
\chi \emi \dead \altc
\Altc{i \in \set{1,\ldots,n}} \phi_i \gc a_i \seqc p_i \altc
\Altc{j \in \set{1,\ldots,m}} \psi_j \gc b_j\;,
\end{ldispl}%
where $n,m \in \Nat$, 
where $\chi \notin [\False]$,
where $\phi_i \notin [\False]$, $a_i \in \Act$, and $p_i \in \cB$ for 
all $i \in \set{1,\ldots,n}$, and
where $\psi_j \notin [\False]$ and $b_j \in \Act$ for 
all $j \in \set{1,\ldots,m}$.
The subterm $\chi$ is called the \emph{root signal} of the basic
term and the subterms $\phi_i \gc a_i \seqc p_i$ and $\psi_j \gc b_j$
are called the \emph{summands} of the basic term.

All closed \ctBPAps\ terms of sort $\Proc$ can be reduced to a basic 
term.
\begin{proposition}[Elimination]
\label{proposition-elim-ctBPAps}
For all closed \ctBPAps\ terms $p$ of sort $\Proc$, there exists a 
$q \in \cB$ such that $p = q$ is derivable from the axioms of \ctBPAps.
\end{proposition}
\begin{proof}
The proof is straightforward by induction on the structure of closed 
term $p$.
If $p$ is of the form $\nex$, $a$, $p' \altc p''$ or $\phi \emi p'$,
then it is trivial to show that there exists a $q \in \cB$ such that 
$p = q$ is derivable from the axioms of \ctBPAps.
If $p$ is of the form $p' \seqc p''$ or $\phi \gc p'$, then it follows
immediately from the induction hypothesis and the following claims:
\begin{itemize}
\item
for all $p,p' \in \cB$, there exists a $p'' \in \cB$ such that 
$p \seqc p' = p''$ is derivable from the axioms of \ctBPAps;
\item
for all $\phi \notin [\False]$ and $p \in \cB$, there exists a 
$p' \in \cB$ such that $\phi \gc p = p'$ is derivable from the axioms of 
\ctBPAps.
\end{itemize}
Both claims are easily proved by induction on the structure of basic
term $p$.
\qed
\end{proof}

\section{Semantics of \boldmath{$\ctBPAps$}}
\label{sect-sem-ctBPAps}

In this section, we present a structural operational semantics of 
\ctBPAps, define a notion of bisimulation equivalence based on this 
semantics, and show that the axioms of \ctBPAps\ are sound and complete 
with respect to this bisimulation equivalence.

We start with the presentation of the structural operational semantics 
of \ctBPAps.
The following transition relations on closed terms of sort $\Proc$ are 
used:
\begin{itemize}
\item 
for each $\ell \in C \x \Act$,
a binary \emph{action step} relation ${} \step{\ell} {}$;
\item 
for each $\ell \in C \x \Act$,
a unary \emph{action termination} relation ${} \term{\ell}$;
\item
for each $\phi \in C$, a unary \emph{signal emission} relation 
$\sgn{\phi}$; \phantom{${} \step{\ell} {}$}
\end{itemize}
where $C$ is the set of all closed terms $\phi$ of sort $\Prop$ such 
that $\phi \notin [\False]$.
We write $\astep{p}{\gact{\phi}{a}}{q}$ instead of 
$\tup{p,q} \in {\step{\tup{\phi,a}}}$, 
$\aterm{p}{\gact{\phi}{a}}$ instead of 
$p \in {\term{\tup{\phi,a}}}$, and
$\rsgn{p} = \phi$ instead of $p \in \sgn{\phi}$.
These relations can be explained as follows:
\begin{itemize}
\item
$\aterm{p}{\gact{\phi}{a}}$: 
$p$ is capable of performing action $a$ under condition $\phi$ and 
then terminating successfully;
\item
$\astep{p}{\gact{\phi}{a}}{q}$: 
$p$ is capable of performing action $a$ under condition $\phi$ and 
then proceeding as $q$;
\item
$\rsgn{p} = \phi$: 
the proposition that holds at the start of $p$ is $\phi$.
\end{itemize}
The structural operational semantics of \ctBPAps\ is described by the 
transition rules given in Table~\ref{sos-ctBPAps}.%
\begin{table}[!tb]
\caption{Transition rules for \ctBPAps}
\label{sos-ctBPAps}
\begin{druletbl}
\Rule
{\phantom{\aterm{a}{\gact{\phi}{a}}}}
{\aterm{a}{\gact{\True}{a}}}
\\
\RuleC
{\aterm{x}{\gact{\phi}{a}},\; \rsgn{x \altc y} = \psi}
{\aterm{x \altc y}{\gact{\phi}{a}}}
{\psi \notin [\False]}
& 
\RuleC
{\aterm{y}{\gact{\phi}{a}},\; \rsgn{x \altc y} = \psi}
{\aterm{x \altc y}{\gact{\phi}{a}}}
{\psi \notin [\False]}
\\
\RuleC
{\astep{x}{\gact{\phi}{a}}{x'},\; \rsgn{x \altc y}  = \psi}
{\astep{x \altc y}{\gact{\phi}{a}}{x'}}
{\psi \notin [\False]}
& 
\RuleC
{\astep{y}{\gact{\phi}{a}}{y'},\; \rsgn{x \altc y} = \psi}
{\astep{x \altc y}{\gact{\phi}{a}}{y'}}
{\psi \notin [\False]}
\\
\RuleC
{\aterm{x}{\gact{\phi}{a}},\; \rsgn{y} = \psi}
{\astep{x \seqc y}{\gact{\phi}{a}}{y}}
{\psi \notin [\False]}
& 
\Rule
{\astep{x}{\gact{\phi}{a}}{x'}}
{\astep{x \seqc y}{\gact{\phi}{a}}{x' \seqc y}}
\\
\RuleC
{\aterm{x}{\gact{\phi}{a}}}
{\aterm{\psi \gc x}{\gact{\phi \CAnd \psi}{a}}}
{\phi \CAnd \psi \notin [\False]}
& 
\RuleC
{\astep{x}{\gact{\phi}{a}}{x'}}
{\astep{\psi \gc x}{\gact{\phi \CAnd \psi}{a}}{x'}}
{\phi \CAnd \psi \notin [\False]}
\\
\RuleC
{\aterm{x}{\gact{\phi}{a}},\; \rsgn{\psi \emi x} = \chi}
{\aterm{\psi \emi x}{\gact{\phi}{a}}}
{\chi \notin [\False]}
& 
\RuleC
{\astep{x}{\gact{\phi}{a}}{x'},\; \rsgn{\psi \emi x} = \chi}
{\astep{\psi \emi x}{\gact{\phi}{a}}{x'}}
{\chi \notin [\False]}
\eqnsep
\Rule
{\phantom{\rsgn{\nex} = \False}}
{\rsgn{\nex} = \False}
\qquad
\Rule
{\phantom{\rsgn{a} = \True}}
{\rsgn{a} = \True}
\\
\Rule
{\rsgn{x} = \phi,\; \rsgn{y} = \psi}
{\rsgn{x \altc y} = \phi \CAnd \psi}
\qquad
\Rule
{\rsgn{x} = \phi}
{\rsgn{x \seqc y} = \phi}
& 
\Rule
{\rsgn{x} = \phi}
{\rsgn{\psi \gc y} = \psi \IImpl \phi}
\qquad
\Rule
{\rsgn{x} = \phi}
{\rsgn{\psi \emi y} = \psi \CAnd \phi}
\end{druletbl}
\end{table}
In this table, $a$ stands for an arbitrary constant from 
$\Act \union \set{\dead}$ and $\phi$, $\psi$, and $\chi$ stand for 
arbitrary closed terms of sort $\Prop$.

A \emph{bisimulation} is a binary relation $R$ on closed \ctBPAps\ terms 
of sort $\Proc$ such that, for all closed \ctBPAps\ terms $p,q$ of sort 
$\Proc$ with $(p,q) \in R$, the following conditions hold:
\begin{itemize}
\item
if $\astep{p}{\gact{\phi}{a}}{p'}$, then, for all valuations $\nu$ with 
$\nu(\rsgn{p}) \neq \false$ and $\nu(\phi) \neq \false$, there exists a
closed term $\psi$ of sort $\Prop$ and a closed term $q'$ of sort 
$\Proc$ such that $\nu(\phi) = \nu(\psi)$, 
$\astep{q}{\gact{\psi}{a}}{q'}$, and $(p',q') \in R$;
\item
if $\astep{q}{\gact{\psi}{a}}{q'}$, then, for all valuations $\nu$ with 
$\nu(\rsgn{q}) \neq \false$ and $\nu(\psi) \neq \false$, there exists a
closed term $\phi$ of sort $\Prop$ and a closed term $p'$ of sort 
$\Proc$ such that $\nu(\psi) = \nu(\phi)$, 
$\astep{p}{\gact{\phi}{a}}{p'}$, and $(p',q') \in R$;
\item
if $\aterm{p}{\gact{\phi}{a}}$, then, for all valuations $\nu$ with 
$\nu(\rsgn{p}) \neq \false$ and $\nu(\phi) \neq \false$, there exists a
closed term $\psi$ of sort $\Prop$ such that $\nu(\phi) = \nu(\psi)$ and 
$\aterm{q}{\gact{\psi}{a}}$;
\item
if $\aterm{q}{\gact{\psi}{a}}$, then, for all valuations $\nu$ with 
$\nu(\rsgn{q}) \neq \false$ and $\nu(\psi) \neq \false$, there exists a
closed term $\phi$ of sort $\Prop$ such that $\nu(\psi) = \nu(\phi)$ and 
$\aterm{p}{\gact{\phi}{a}}$;
\item
if $\rsgn{p} = \phi$, then there exists a closed term $\psi$ of sort 
$\Prop$ such that $\rsgn{q} = \psi$ and $\phi \LEqv \psi$;
\item
if $\rsgn{q} = \psi$, then there exists a closed term $\phi$ of sort 
$\Prop$ such that $\rsgn{p} = \phi$ and $\psi \LEqv \phi$.
\end{itemize}
Two closed \ctBPAps\ terms $p,q$ of sort $\Proc$ are \emph{bisimulation 
equivalent}, written $p \bisim q$, if there exists a bisimulation $R$
such that $(p,q) \in R$.
Let $R$ be a bisimulation such that $(p,q) \in R$.
Then we say that $R$ is a bisimulation \emph{witnessing} $p \bisim q$.

Henceforth, we will loosely say that a relation contains all closed 
substitution instances of an equation if it contains all pairs $(t,t')$ 
such that $t = t'$ is a closed substitution instance of the equation.

Because a transition on one side may be simulated by a set of 
transitions on the other side, a bisimulation as defined above is called
a \emph{splitting} bisimulation in~\cite{BM05a}.

Bisimulation equivalence is a congruence with respect to the operators
of \ctBPAps.
\begin{proposition}[Congruence]
\label{proposition-congr-ctBPAps}
For all closed \ctBPAps\ terms $p,q,p',q'$ of sort $\Proc$ and closed 
\ctBPAps\ terms $\phi$ of sort $\Prop$, $p \bisim q$ and $p' \bisim q'$
imply $p \altc p' \bisim q \altc q'$, $p \seqc p' \bisim q \seqc q'$, 
$\phi \gc p \bisim \phi \gc q$, and $\phi \emi p \bisim \phi \emi q$.
\end{proposition}
\begin{proof}
We can reformulate the transition rules such that:
\begin{itemize}
\item
bisimulation equivalence based on the reformulated transition rules 
according to the standard definition of bisimulation equivalence 
coincides with bisimulation equivalence based on the original transition 
rules according to the definition of bisimulation equivalence given 
above;
\item
the reformulated transition rules make up a complete transition system 
specification in panth format.
\end{itemize}
The reformulation goes like the one for the transition rules for BPAps
outlined in~\cite{BB94b}.
The proposition follows now immediately from the well-known result that
bisimulation equivalence according to the standard definition of 
bisimulation equivalence is a congruence if the transition rules 
concerned make up a complete transition system specification in panth 
format (see e.g.~\cite{FG96a}).
\qed
\end{proof}
%
The underlying idea of the reformulation referred to above is that we 
replace each transition \smash{$\astep{p}{\gact{\phi}{a}}{p'}$} by a 
transition \smash{$\astep{p}{\gact{\nu}{a}}{p'}$} for each valuation 
$\nu$ such that $\nu(\phi) \neq \false$, and likewise 
\smash{$\aterm{p}{\gact{\phi}{a}}$} and $\rsgn{p} = \phi$.
Thus, in a bisimulation, a transition on one side must be simulated by a
single transition on the other side.
We did not present the reformulated structural operational semantics
in this paper because it is, in our opinion, intuitively less appealing.  

\ctBPAps\ is sound with respect to $\bisim$ for equations between closed
terms.
\begin{theorem}[Soundness]
\label{theorem-soundness-ctBPAps}
For all closed \ctBPAps\ terms $p,q$ of sort $\Proc$, $p = q$ is 
derivable from the axioms of \ctBPAps\ only if $p \bisim q$.
\end{theorem}
\begin{proof}
Because of Proposition~\ref{proposition-congr-ctBPAps}, it is sufficient
to prove the theorem for all closed substitution instances of each axiom 
of \ctBPAps.

For each axiom, we can construct a bisimulation $R$ witnessing 
$p \bisim q$ for all closed substitution instances $p = q$ of the axiom 
as follows:
\begin{itemize}
\item
in the case of A1--A4 and A6, we take the relation $R$ that consists of 
all closed substitution instances of the axiom concerned and the 
equation $x = x$;
\item
in the case of A5, we take the relation $R$ that consists of all closed 
substitution instances of A5, SE5, and the equation $x = x$;
\item
in the case of A7, NE1--NE3, GC2--GC3, and SE2--SE3, we take the 
relation $R$ that consists of all closed substitution instances of the 
axiom concerned;
\item
in the case of GC1, GC4--GC7, SE1, and SE4--SE8, we take the relation 
$R$ that consists of all closed substitution instances of the axiom 
concerned and the equation $x = x$.
\end{itemize}
The laws from property~(8) of LP$^{\IImpl,\False}$ mentioned in 
Section~\ref{sect-LP-iimpl-false} are needed to check that these 
relations are witnessing ones.
\qed
\end{proof}
The proof of Theorem~\ref{theorem-soundness-ctBPAps} goes along the same 
line as the soundness proof for BPAps outlined in~\cite{BB94b}.
The laws from property~(8) of LP$^{\IImpl,\False}$ mentioned in 
Section~\ref{sect-LP-iimpl-false} are laws that LP$^{\IImpl,\False}$ has 
in common with classical propositional logic.
They are needed in the soundness proof for BPAps as well, but their use 
is left implicit in the proof outline given in~\cite{BB94b}.

\ctBPAps\ is complete with respect to $\bisim$ for equations between
closed terms.

\begin{theorem}[Completeness]
\label{theorem-completeness-ctBPAps}
For all closed \ctBPAps\ terms $p,q$ of sort $\Proc$, $p = q$ is 
derivable from the axioms of \ctBPAps\ if $p \bisim q$.
\end{theorem}
\begin{proof}
By Proposition~\ref{proposition-elim-ctBPAps} and 
Theorem~\ref{theorem-soundness-ctBPAps}, it is sufficient to prove the
theorem for basic terms $p$ and $q$.

For $p,p' \in \cB$, $p'$ is called a \emph{basic subterm} of $p$ if 
$p' \equiv p$ or there exists an $a \in \Act$ such that $a \seqc p'$ is 
a subterm of $p$.

We introduce a reduction relation $\tsred$ on $\cB$.
The one-step reduction relation $\tsredi$ on $\cB$ is inductively 
defined as follows:
\begin{itemize}
\item
if $p'$ is a basic subterm of $p$ and $q'$ occurs twice as summand in 
$p'$, then $p \tsredi r$ where $r$ is $p$ with one occurrence of $q'$ 
removed;
\item
if $p'$ is a basic subterm of $p$ and both $\phi \gc a \seqc q'$ and
$\psi \gc a \seqc q'$ occur as summand in $p'$, then $p \tsredi r$ where 
$r$ is $p$ with the occurrence of $\phi \gc a \seqc q'$ replaced by
$\phi \COr \psi \gc a \seqc q'$ and the occurrence of 
$\psi \gc a \seqc q'$ removed;
\item
if $p'$ is a basic subterm of $p$ and both $\phi \gc a$ and $\psi \gc a$ 
occur as summand in $p'$, then $p \tsredi r$ where $r$ is $p$ with the 
occurrence of $\phi \gc a$ replaced by $\phi \COr \psi \gc a$ and the 
occurrence of $\psi \gc a$ removed.
\end{itemize}
The one-step reductions correspond to sharing of double states and 
joining of transitions as in~\cite{BM05d}.
The reduction relation $\tsred$ is the reflexive and transitive closure 
of $\tsredi$, and the conversion relation $\tscon$ is the reflexive and 
transitive closure of $\tsredi \union \tsredi^{-1}$.

The following are important properties of $\tsredi$:
\begin{enumerate}
\item[(1)]
$\tsred$ is strongly normalizing;
\item[(2)]
for all $p,q \in \cB$, $p \tsred q$ only if $p \bisim q$;
\item[(3)]
for all $p,q \in \cB$ that are in normal form, $p \bisim q$ only if  
$p = q$ is derivable from axioms A1 and A2;
\item[(4)]
for all $p,q \in \cB$, $p \tsredi q$ only if $p = q$ is derivable from 
the axioms of \ctBPAps.
\end{enumerate}
Verifying properties~(1), (2), and~(4) is trivial.
Property~(3) can be verified by proving it, simultaneously with the 
property
\begin{enumerate}
\item[]
for all $p \in \cB$ that are in normal form, any bisimulation between 
$p$ and itself is the identity relation,
\end{enumerate}
by induction on the number of occurrences of a constant from $\Act$ in 
$p$ and $q$.
The proof is similar to the proof of Theorem~2.12 from~\cite{BK85b}, but
easier.

From properties~(1), (2) and~(3), it follows immediately that, for all 
$p,q \in \cB$, $p \bisim q$ iff $p \tscon q$.
From this and property~(4), it follows immediately that, for all 
$p,q \in \cB$, $p \bisim q$ only if $p = q$ is derivable from the 
axioms of \ctBPAps.
\qed
\end{proof}

\section{Contradiction-Tolerant ACP with Propositional Signals}
\label{sect-ctACPps}

In this section, we present the contradiction-tolerant version of ACPps.
This version, which is called \ctACPps, is an extension of \ctBPAps\
that supports parallelism and communication.

In \ctACPps, just as in \ctBPAps, it is assumed that a fixed but 
arbitrary finite set $\Act$ of actions, with $\dead \not\in \Act$, and a 
fixed but arbitrary finite set $\AProp$ of atomic propositions have been 
given.
In \ctACPps, it is further assumed that a fixed but arbitrary commutative 
and associative \emph{communication} function 
$\funct{\commm}{\Actd \x \Actd}{\Actd}$, such that 
$\dead \commm a = \dead$ for all $a \in \Actd$, has been given.
The function $\commm$ is regarded to give the result of synchronously
performing any two actions for which this is possible, and to be $\dead$
otherwise.

The algebraic theory \ctACPps\ has the sorts, constants and operators 
of \ctBPAps\ and in addition the following operators:
\begin{itemize}
\item
the binary \emph{parallel composition} operator 
$\funct{\parc}{\Proc \x \Proc}{\Proc}$;
\item
the binary \emph{left merge} operator 
$\funct{\leftm}{\Proc \x \Proc}{\Proc}$;
\item
the binary \emph{communication merge} operator 
$\funct{\commm}{\Proc \x \Proc}{\Proc}$;
\item
for each $H \subseteq \Act$, 
the unary \emph{encapsulation} operator 
$\funct{\encap{H}}{\Proc}{\Proc}$.
\end{itemize}
We use infix notation for the additional binary operators as well.

The constants and operators of \ctACPps\ to build terms of sort $\Proc$
are the constants and operators of \ACP\ and additionally the guarded 
command operator and the signal emission operator.

Let $p$ and $q$ be closed terms of sort $\Proc$.
Intuitively, the additional operators can be explained as follows:
\begin{itemize}
\item
$p \parc q$ behaves as the process that proceeds with $p$ and $q$ in 
parallel, the proposition that holds at the start of $p \parc q$ is the 
conjunction of the propositions that hold at the start of $p$ and $q$;
\item
$p \leftm q$ behaves the same as $p \parc q$, except that it starts 
with performing an action of $p$, the proposition that holds at the 
start of $p \leftm q$ is the conjunction of the propositions that hold 
at the start of $p$ and $q$;
\item
$p \commm q$ behaves the same as $p \parc q$, except that it starts 
with performing an action of $p$ and an action of $q$ synchronously, the 
proposition that holds at the start of $p \commm q$ is the conjunction 
of the propositions that hold at the start of $p$ and~$q$;
\item
$\encap{H}(p)$ behaves the same as $p$, except that the actions in $H$
are blocked, the proposition that holds at the start of $\encap{H}(p)$ 
is the proposition that holds at the start of $p$.
\end{itemize}

The axioms of \ctACPps\ are the axioms of \ctBPAps\ and the additional
axioms given in Table~\ref{axioms-ctACPps}.%
\begin{table}[!tb]
\caption{Additional axioms for \ctACPps}
\label{axioms-ctACPps}
\begin{eqntbl}
\begin{axcol}
x \parc y =
     x \leftm y \altc y \leftm x \altc x \commm y         & \ax{CM1} \\
a \leftm x = a \seqc x \altc \encap{\Act}(x)              & \ax{CM2S}\\
a \seqc x \leftm y = a \seqc (x \parc y) \altc \encap{\Act}(y)
                                                          & \ax{CM3S}\\
(x \altc y) \leftm z = x \leftm z \altc y \leftm z        & \ax{CM4} \\
a \seqc x \commm b = (a \commm b) \seqc x                 & \ax{CM5} \\
a \commm b \seqc x = (a \commm b) \seqc x                 & \ax{CM6} \\
a \seqc x \commm b \seqc y = (a \commm b) \seqc (x \parc y)
                                                          & \ax{CM7} \\
(x \altc y) \commm z = x \commm z \altc y \commm z        & \ax{CM8} \\
x \commm (y \altc z) = x \commm y \altc x \commm z        & \ax{CM9} 
\\ {} \\
(\phi \gc x) \leftm y = \phi \gc (x \leftm y) \altc \encap{\Act}(y)
                                                          & \ax{GC8S}\\
(\phi \gc x) \commm y = \phi \gc (x \commm y) \altc \encap{\Act}(y)
                                                          & \ax{GC9S}\\
x \commm (\phi \gc y) = \phi \gc (x \commm y) \altc \encap{\Act}(x)
                                                          & \ax{GC10S}\\
\encap{H}(\phi \gc x) = \phi \gc \encap{H}(x)             & \ax{GC11}
\end{axcol}
\qquad
\begin{axcol}
a \commm b = b \commm a                                   & \ax{C1} \\
(a \commm b) \commm c = a \commm (b \commm c)             & \ax{C2} \\
\dead \commm a = \dead                                    & \ax{C3} 
\\ {} \\ {} \\
\encap{H}(a) = a            \hfill \mif a \not\in H       & \ax{D1} \\
\encap{H}(a) = \dead            \hfill \mif a \in H       & \ax{D2} \\
\encap{H}(x \altc y) =
                    \encap{H}(x) \altc \encap{H}(y)       & \ax{D3} \\
\encap{H}(x \seqc y) =
                    \encap{H}(x) \seqc \encap{H}(y)       & \ax{D4} 
\\ {} \\
(\phi \emi x) \leftm y = \phi \emi (x \leftm y)           & \ax{SE9}\\
(\phi \emi x) \commm y = \phi \emi (x \commm y)           & \ax{SE10}\\
x \commm (\phi \emi y) = \phi \emi (x \commm y)           & \ax{SE11}\\
\encap{H}(\phi \emi x) = \phi \emi \encap{H}(x)           & \ax{SE12}
\end{axcol}
\end{eqntbl}
\end{table}
In this table, $a,b,c$ stand for arbitrary constants from 
$\Act \union \set{\dead}$ and $\phi$ stands for an arbitrary closed term 
of sort $\Prop$.
A1--A7, CM1--CM9 with CM1S and CM2S replaced by $a \leftm x = a \seqc x$ 
and $a \seqc x \leftm y = a \seqc (x \parc y)$, C1--C3, and D1--D4 are 
the axioms of \ACP\ (see e.g.~\cite{BW90}).
GC11 and SE9--SE12 have been taken from~\cite{BB94b} and GC9S and GC10S 
have been taken from~\cite{BB94b} with subterms of the form 
$\rsgn{x} \emi \dead$ replaced by $\encap{\Act}(x)$.
CM2S, CM3S and GC8S differ really from the corresponding axioms
in~\cite{BB94b} due to the choice of having as the proposition that 
holds at the start of the left merge of two processes, as in the case of 
the communication merge, always the conjunction of the propositions that
hold at the start of the two processes.

The following equations are among the equations derivable from the 
axioms of \ctACPps:
\begin{ldispl}
(\phi \emi x) \parc (\psi \emi y) = 
(\phi \CAnd \psi ) \emi (x \parc y)\;,
\\
\hfill x \parc \nex = \nex\;, \hfill \nex \parc x = \nex\;. \hfill
\end{ldispl}%

Let $\phi$ be a closed term of sort $\Prop$ such that
not $\pEnt \phi \,\Iff\, \False$ and
not $\pEnt \Not \phi \,\Iff\, \False$.
Then, because not 
$\pEnt \phi \CAnd \Not \phi \,\Iff\, \False$, 
we have that
$a \seqc (\phi \emi x \parc \Not \phi \emi y) =
 a \seqc (\False \emi (x \parc y)) = \dead$,
which is derivable from the axioms of ACPps, is not 
derivable from the axioms of \ctACPps. 
This shows the main difference between \ctACPps\ and ACPps: the 
par\-allel composition of two processes of which the propositions 
that hold at the start of them are contradictory does not lead to an 
inconsistency in \ctACPps, whereas it does lead to an inconsistency in 
ACPps.
This is why \ctACPps\ is called the contradiction-tolerant version of
ACPps.

Let $\phi$ be a closed term of sort $\Prop$ such that
not $\pEnt \phi \,\Iff\, \False$ and
not $\pEnt \Not \phi \,\Iff\, \False$.
Assume that $b \commm c = d$.
Then we can derive 
$a \seqc (\phi \emi b \parc \Not \phi \emi c) =
 a \seqc ((\phi \CAnd \Not \phi) \emi (b \parc c)) = 
 a \seqc 
   ((\phi \CAnd \Not \phi) \emi (b \seqc c \altc c \seqc b \altc d)) =
 \dead$ 
from the axioms of BPAps because, in the case of BPAps, 
$a \seqc (\phi \emi b \parc \Not \phi \emi c)$ 
is not capable of doing anything.
We can only derive
$a \seqc (\phi \emi b \parc \Not \phi \emi c) =
 a \seqc ((\phi \CAnd \Not \phi) \emi (b \parc c)) = 
 a \seqc 
   ((\phi \CAnd \Not \phi) \emi (b \seqc c \altc c \seqc b \altc d))$  
from the axioms of \ctBPAps\ because, in the case of \ctBPAps, 
$a \seqc (\phi \emi b \parc \Not \phi \emi c)$ 
is capable of first performing $a$ and next either performing $b$ and 
$c$ in either order and after that terminating successfully or 
performing $d$ and after that terminating successfully --- although the 
proposition that holds at the start of the process that remains after 
performing $a$ is the contradiction $\phi \CAnd \Not \phi$.

Let $\phi$ be a closed term of sort $\Prop$ such that
not $\pEnt \phi \,\Iff\, \False$ and
not $\pEnt \Not \phi \,\Iff\, \False$.
Then, because
$\pEnt \Cons \phi \CAnd \phi \CAnd \Not \phi \,\Iff\, \False$,  
we have that
$a \seqc (\Cons \phi \emi (\phi \emi x \parc \Not \phi \emi y)) =
 a \seqc (\False \emi (x \parc y)) = \dead$
is derivable from the axioms of \ctACPps.
This shows that it can be enforced by means of a consistency proposition
($\Cons \phi$) that the parallel composition of two processes of which 
the propositions that hold at the start of them are contradictory leads 
to an inconsistency in \ctACPps.

All closed \ctACPps\ terms of sort $\Proc$ can be reduced to a basic 
term.
\begin{proposition}[Elimination]
\label{proposition-elim-ctACPps}
For all closed \ctACPps\ terms $p$ of sort $\Proc$, there exists a 
$q \in \cB$ such that $p = q$ is derivable from the axioms of \ctACPps.
\end{proposition}
\begin{proof}
The proof is straightforward by induction on the structure of closed 
term $p$.
If $p$ is of the form $\nex$, $a$, $p' \altc p''$, $p' \seqc p''$, 
$\phi \gc p'$ or $\phi \emi p'$, then it follows immediately from
the induction hypothesis and Proposition~\ref{proposition-elim-ctBPAps}  
that there exists a $q \in \cB$ such that $p = q$ is derivable from the 
axioms of \ctACPps.
If $p$ is of the form $p' \parc p''$, $p' \leftm p''$, $p' \commm p''$ 
or $\encap{H}(p')$, then it follows immediately from the induction 
hypothesis and claims similar to the ones from the proof of 
Proposition~\ref{proposition-elim-ctBPAps}.
The claims concerning $\parc$, $\leftm$, and $\commm$ are easily proved 
simultaneously by structural induction.
The claim concerning $\encap{H}$ is easily proved by structural 
induction.
\qed
\end{proof}

\section{Semantics of \boldmath{$\ctACPps$}}
\label{sect-sem-ctACPps}

In this section, we present a structural operational semantics of 
\ctACPps\ and show that the axioms of \ctACPps\ are sound and complete 
with respect to this bisimulation equivalence.

We start with the presentation of the structural operational semantics 
of \ctACPps.
The structural operational semantics of \ctACPps\ is described by the
transition rules for \ctBPAps\ and the additional transition rules given 
in Table~\ref{sos-ctACPps}.%
\renewcommand{\textfraction}{0} \renewcommand{\topfraction}{1}%
\begin{table}[!p]
\caption{Additional transition rules for \ctACPps}
\label{sos-ctACPps}
\begin{ruletbl}
\RuleC
{\aterm{x}{\gact{\phi}{a}},\; 
 \rsgn{x \parc y} = \psi,\; \rsgn{y} = \chi}
{\astep{x \parc y}{\gact{\phi}{a}}{y}}
{\psi, \chi \notin [\False]}
\\
\RuleC
{\aterm{y}{\gact{\phi}{a}},\;
 \rsgn{x \parc y} = \psi,\; \rsgn{x} = \chi}
{\astep{x \parc y}{\gact{\phi}{a}}{x}}
{\psi, \chi \notin [\False]}
\\
\RuleC
{\astep{x}{\gact{\phi}{a}}{x'},\;
 \rsgn{x \parc y} = \psi,\; \rsgn{x' \parc y} = \chi}
{\astep{x \parc y}{\gact{\phi}{a}}{x' \parc y}}
{\psi, \chi \notin [\False]}
\\
\RuleC
{\astep{y}{\gact{\phi}{a}}{y'},\;
 \rsgn{x \parc y} = \psi,\; \rsgn{x \parc y'} = \chi}
{\astep{x \parc y}{\gact{\phi}{a}}{x \parc y'}}
{\psi, \chi \notin [\False]}
\\
\RuleC
{\aterm{x}{\gact{\phi}{a}},\; \aterm{y}{\gact{\psi}{b}},\;
 \rsgn{x \parc y} = \chi}
{\aterm{x \parc y}{\gact{\phi \CAnd \psi}{c}}}
{a \commm b = c,\; \phi \CAnd \psi, \chi \notin [\False]}
\\
\RuleC
{\aterm{x}{\gact{\phi}{a}},\; \astep{y}{\gact{\psi}{b}}{y'},\;
 \rsgn{x \parc y} = \chi}
{\astep{x \parc y}{\gact{\phi \CAnd \psi}{c}}{y'}}
{a \commm b = c,\; \phi \CAnd \psi, \chi \notin [\False]}
\\
\RuleC
{\astep{x}{\gact{\phi}{a}}{x'},\; \aterm{y}{\gact{\psi}{b}},\;
 \rsgn{x \parc y} = \chi}
{\astep{x \parc y}{\gact{\phi \CAnd \psi}{c}}{x'}}
{a \commm b = c,\; \phi \CAnd \psi, \chi \notin [\False]}
\\
\RuleC
{\astep{x}{\gact{\phi}{a}}{x'},\; \astep{y}{\gact{\psi}{b}}{y'},\;
 \rsgn{x \parc y} = \chi,\; \rsgn{x' \parc y'} = \chi'}
{\astep{x \parc y}{\gact{\phi \CAnd \psi}{c}}{x' \parc y'}}
{a \commm b = c,\; \phi \CAnd \psi, \chi, \chi' \notin [\False]}
\\
\RuleC
{\aterm{x}{\gact{\phi}{a}},\;
 \rsgn{x \leftm y} = \psi,\; \rsgn{y} = \chi}
{\astep{x \leftm y}{\gact{\phi}{a}}{y}}
{\psi, \chi \notin [\False]}
\\
\RuleC
{\astep{x}{\gact{\phi}{a}}{x'},\;
 \rsgn{x \leftm y} = \psi,\; \rsgn{x' \parc y} = \chi}
{\astep{x \leftm y}{\gact{\phi}{a}}{x' \parc y}}
{\psi, \chi \notin [\False]}
\\
\RuleC
{\aterm{x}{\gact{\phi}{a}},\; \aterm{y}{\gact{\psi}{b}},\;
 \rsgn{x \commm y} = \chi}
{\aterm{x \commm y}{\gact{\phi \CAnd \psi}{c}}}
{a \commm b = c,\; \phi \CAnd \psi, \chi \notin [\False]}
\\
\RuleC
{\aterm{x}{\gact{\phi}{a}},\; \astep{y}{\gact{\psi}{b}}{y'},\;
 \rsgn{x \commm y} = \chi}
{\astep{x \commm y}{\gact{\phi \CAnd \psi}{c}}{y'}}
{a \commm b = c,\; \phi \CAnd \psi, \chi \notin [\False]}
\\
\RuleC
{\astep{x}{\gact{\phi}{a}}{x'},\; \aterm{y}{\gact{\psi}{b}},\;
 \rsgn{x \commm y} = \chi}
{\astep{x \commm y}{\gact{\phi \CAnd \psi}{c}}{x'}}
{a \commm b = c,\; \phi \CAnd \psi, \chi \notin [\False]}
\\
\RuleC
{\astep{x}{\gact{\phi}{a}}{x'},\; \astep{y}{\gact{\psi}{b}}{y'},\;
 \rsgn{x \commm y}  = \chi,\; \rsgn{x' \parc y'}  = \chi'}
{\astep{x \commm y}{\gact{\phi \CAnd \psi}{c}}{x' \parc y'}}
{a \commm b = c,\; \phi \CAnd \psi, \chi, \chi' \notin [\False]}
\\
\RuleC
{\aterm{x}{\gact{\phi}{a}}}
{\aterm{\encap{H}(x)}{\gact{\phi}{a}}}
{a \not\in H}
\qquad
\RuleC
{\astep{x}{\gact{\phi}{a}}{x'}}
{\astep{\encap{H}(x)}{\gact{\phi}{a}}{\encap{H}(x')}}
{a \not\in H}
\eqnsep
\Rule
{\rsgn{x} = \phi,\; \rsgn{y} = \psi}
{\rsgn{x \parc y} = \phi \CAnd \psi}
\qquad
\Rule
{\rsgn{x} = \phi,\; \rsgn{y} = \psi}
{\rsgn{x \leftm y} = \phi \CAnd \psi}
\qquad
\Rule
{\rsgn{x} = \phi,\; \rsgn{y} = \psi}
{\rsgn{x \commm y} = \phi \CAnd \psi}
\qquad
\Rule
{\rsgn{x} = \phi}
{\rsgn{\encap{H}(x)} = \phi}
\end{ruletbl}
\end{table}
In these tables, $a$, $b$, and $c$ stand for arbitrary constants from 
$\Act \union \set{\dead}$ and $\phi$, $\psi$, $\chi$, and $\chi'$ stand 
for arbitrary closed terms of sort $\Prop$.

In Sections~\ref{sect-ctBPAps} and~\ref{sect-ctACPps}, we have touched
upon the main difference between \ctACPps\ and ACPps: the alternative 
and parallel composition of two processes of which the propositions that 
hold at the start of them are contradictory does not lead to an 
inconsistency in \ctACPps, whereas it does lead to an inconsistency in 
ACPps.
However, the transition rules for \ctACPps\ and ACPps seem to be the 
same.
The difference is fully accounted for by the fact that $[\False]$, the
equivalence class of $\False$ modulo logical equivalence, contains in 
the case of LP$^{\IImpl,\False}$ only propositions of the form 
$\phi \CAnd \Not \phi$ with $\phi$ such that either $\phi \LEqv \False$ 
or $\Not \phi \LEqv \False$, whereas it contains in the case of 
classical propositional logic all propositions of the form 
$\phi \CAnd \Not \phi$.

By this fact, in the case of \ctACPps, 
$a \seqc (\phi \emi b \parc \Not \phi \emi c)$ from the example 
preceding Proposition~\ref{proposition-elim-ctACPps} is capable of 
first performing $a$ and next either performing $b$ and $c$ in either 
order and after that terminating successfully or performing $d$ and 
after that terminating successfully --- although the proposition that 
holds at the start of the process that remains after performing $a$ is 
the contradiction $\phi \CAnd \Not \phi$ --- and, in 
the case of ACPps, it is not capable of doing anything.

Bisimulation equivalence is a congruence with respect to the operators 
of \ctACPps.
\begin{proposition}[Congruence]
\label{proposition-congr-ctACPps}
For all closed \ctACPps\ terms $p,q,p',q'$ of sort $\Proc$ and closed 
\ctACPps\ terms $\phi$ of sort $\Prop$, $p \bisim q$ and $p' \bisim q'$
imply $p \altc p' \bisim q \altc q'$, $p \seqc p' \bisim q \seqc q'$, 
$\phi \gc p \bisim \phi \gc q$, $\phi \emi p \bisim \phi \emi q$,
$p \parc p' \bisim q \parc q'$, $p \leftm p' \bisim q \leftm q'$, 
$p \commm p' \bisim q \commm q'$, and 
$\encap{H}(p) \bisim \encap{H}(q)$. 
\end{proposition}
\begin{proof}
The proof goes along the same line as the proof of 
Proposition~\ref{proposition-congr-ctBPAps}.
\qed
\end{proof}

\ctACPps\ is sound with respect to $\bisim$ for equations between closed 
terms.
\begin{theorem}[Soundness]
\label{theorem-soundness-ctACPps}
For all closed \ctACPps\ terms $p,q$ of sort $\Proc$, $p = q$ is 
derivable from the axioms of \ctACPps\ only if $p \bisim q$.
\end{theorem}
\begin{proof}
Because of Proposition~\ref{proposition-congr-ctACPps}, it is sufficient
to prove the theorem for all closed substitution instances of each axiom 
of \ctACPps.

For each axiom, we can construct a bisimulation $R$ witnessing 
$p \bisim q$ for all closed substitution instances $p = q$ of the axiom 
as follows:
\begin{itemize}
\item
in the case of the axioms of \ctBPAps,we take the same relation as in 
the proof of Theorem~\ref{theorem-soundness-ctBPAps};
\item
in the case of CM1, we take the relation $R$ that consists of all closed 
substitution instances of CM1, the equation $x \parc y = y \parc x$, and 
the equation $x = x$;
\item
in the case of CM2S--CM9, we take the relation $R$ that consists of all 
closed substitution instances of the axiom concerned and the equation 
$x = x$;
\item
in the case of C1--C3 and D1--D2, we take the relation $R$ that consists 
of all closed substitution instances of the axiom concerned;
\item
in the case of D3--D4, GC8S--GC11, and SE9--SE12, we take the relation 
$R$ that consists of all closed substitution instances of the axiom 
concerned and the equation $x = x$.
\end{itemize}
The laws from property~(8) of LP$^{\IImpl,\False}$ mentioned in 
Section~\ref{sect-LP-iimpl-false} are needed to check that these 
relations are witnessing ones.
\qed
\end{proof}

\ctACPps\ is complete with respect to $\bisim$ for equations between
closed terms.
\begin{theorem}[Completeness]
\label{theorem-completeness-ctACPps}
For all closed \ctACPps\ terms $p,q$ of sort $\Proc$, $p = q$ is 
derivable from the axioms of \ctACPps\ if $p \bisim q$.
\end{theorem}
\begin{proof}
We have that the axioms of \ctBPAps\ are complete with respect to 
$\bisim$ (Theorem~\ref{theorem-completeness-ctBPAps}), the axioms of 
\ctACPps\ are sound with respect to $\bisim$ 
(Theorem~\ref{theorem-soundness-ctACPps}), and for each closed \ctACPps\ 
term $p$ of sort $\Proc$, there exists a closed \ctBPAps\ term $q$ such 
that $p = q$ is derivable from the axioms of \ctACPps\
(Proposition~\ref{proposition-elim-ctACPps}).
By Theorem 3.14 from~\cite{Ver94a}, the result immediately follows from
this and the claim that the set of transition rules for \ctACPps\ is an 
operational conservative extension of the set of transition rules for 
\ctBPAps.

This claim can easily be proved if we reformulate the transition rules 
for \ctACPps\ in the same way as the transition rules for \ctBPAps\ have 
been reformulated to prove Proposition~\ref{proposition-congr-ctBPAps}.
The operational conservativity can then easily be proved by verifying 
that the reformulated transition rules for \ctACPps\ makes up a complete 
transition system specification, the reformulated transition rules for 
\ctBPAps --- which are included in the reformulated transition rules for 
\ctACPps --- are source-dependent, and the additional transition rules 
have fresh sources (see e.g.~\cite{FV98a}).
\qed
\end{proof}

\section{State Operators}
\label{sect-ctACPps+SO}

In this section, we extend \ctACPps\ with state operators.
The resulting theory is called \ctACPps\textup{{+}SO}.
The state operators introduced here generalize the state operators added 
to \ACP\ in~\cite{BB88}.

The state operators from~\cite{BB88} were introduced to make it easy to 
represent the execution of a process in a state. 
The basic idea was that the execution of an action in a state has effect
on the state, i.e.\ it causes a change of state. 
Moreover, there is an action left when an action is executed in a state.
The main difference between the original state operators and the state 
operators introduced here is that, in the case of the latter, the state 
in which a process is executed determines the proposition that holds at 
its start.
Thus, one application of a state operator may replace many applications 
of the signal emission operator. 

It is assumed that a fixed but arbitrary set $S$ of \emph{states} has
been given, together with functions 
$\funct{\act}{\Act \x S}{\Actd}$, $\funct{\eff}{\Act \x S}{S}$, and
$\funct{\sig}{S}{B}$, where $B$ is the set of all closed terms $\phi$ 
of sort $\Prop$.

For each $s \in S$, we add a unary \emph{state} operator 
$\funct{\state{s}}{\Proc}{\Proc}$ to the operators of \ctACPps.

The state operator $\state{s}$ allows, given the above-mentioned 
functions, processes to be executed in a state.
Let $p$ be a closed term of sort $\Proc$.
Then $\state{s}(p)$ is the process $p$ executed in state $s$.
The function $\act$ gives, for each action $a$ and state $s$,
the action that results from executing $a$ in state $s$.
The function $\eff$ gives, for each action $a$ and state $s$,
the state that results from executing $a$ in state $s$.
The function $\sig$ gives, for each state $s$, the proposition that
holds at the start of any process executed in state $s$.

The additional axioms for $\state{s}$, where $s \in S$, are given in 
Table~\ref{axioms-SO}.
\begin{table}[!tb]
\caption{Axioms for state operators}
\label{axioms-SO}
\begin{eqntbl} 
\begin{axcol}
\state{s}(a) = \sig(s) \emi \act(a,s)                      & \ax{SO1} \\
\state{s}(a \seqc x) = \sig(s) \emi \act(a,s) \seqc \state{\eff(a,s)}(x) 
                                                           & \ax{SO2} \\
\state{s}(x \altc y) = \state{s}(x) \altc \state{s}(y)     & \ax{SO3} \\
\state{s}(\phi \gc x) = \sig(s) \emi (\phi \gc \state{s}(x))
                                                           & \ax{SO4} \\
\state{s}(\phi \emi x) = \phi \emi \state{s}(x)            & \ax{SO5} 
\end{axcol}
\end{eqntbl}
\end{table}
In this table, $a$ stands for an arbitrary constant from 
$\Act \union \set{\dead}$ and $\phi$ stands for an arbitrary closed term 
of sort $\Prop$.
SO1--SO5 have been taken from~\cite{BB94b}.

The following equations are among the equations derivable from the axioms of 
\ctACPps{+}SO:
\begin{ldispl}
\state{s}(\nex) = \nex\;, \qquad \state{s}(\dead) = \sig(s) \emi \dead\;.
\end{ldispl}%

All closed \ctACPps{+}SO terms of sort $\Proc$ can be reduced to a basic 
term.
\begin{proposition}[Elimination]
\label{proposition-elim-ctACPps+SO}
For all \ctACPps\textup{{+}SO} closed terms $p$ of sort $\Proc$, there 
exists a $q \in \cB$ such that $p = q$ is derivable from the axioms of 
\ctACPps\textup{{+}SO}.
\end{proposition}
\begin{proof}
The proof goes along the same line as the proof of 
Proposition~\ref{proposition-elim-ctBPAps}.
\qed
\end{proof}

The additional transition rules for the state operators are given in 
Table~\ref{sos-SO}.
\begin{table}[!tb]
\caption{Transition rules for state operators}
\label{sos-SO}
\begin{ruletbl}
\RuleC
{\aterm{x}{\gact{\phi}{a}},\; \rsgn{\state{s}(x)} = \psi}
{\aterm{\state{s}(x)}{\gact{\phi}{\act(a,s)}}}
{\act(a,s) \neq \dead,\; \psi \notin [\False]}
\\
\RuleC
{\astep{x}{\gact{\phi}{a}}{x'},\;
 \rsgn{\state{s}(x)} = \psi,\; \rsgn{\state{\eff(a,s)}(x')} = \chi}
{\astep{\state{s}(x)}{\gact{\phi}{\act(a,s)}}{\state{\eff(a,s)}(x')}}
{\act(a,s) \neq \dead,\; \psi, \chi \notin [\False]}
\\
\RuleC
{\rsgn{x} = \phi}
{\rsgn{\state{s}(x)} = \phi \CAnd \psi}
{\sig(s) = \psi}
\end{ruletbl}
\end{table}
In this table, $a$ stands for an arbitrary constant from 
$\Act \union \set{\dead}$ and $\phi$ stands for an arbitrary closed term 
of sort $\Prop$.

Bisimulation equivalence is a congruence with respect to  the operators 
of \ctACPps{+}SO.
\begin{proposition}[Congruence]
\label{proposition-congr-ctACPps+SO}
For all closed \ctACPps\textup{{+}SO} terms $p,q,p',q'$ of sort $\Proc$ 
and closed \ctACPps\textup{{+}SO} terms $\phi$ of sort $\Prop$, 
$p \bisim q$ and $p' \bisim q'$ imply $p \altc p' \bisim q \altc q'$, 
$p \seqc p' \bisim q \seqc q'$, $\phi \gc p \bisim \phi \gc q$, 
$\phi \emi p \bisim \phi \emi q$, $p \parc p' \bisim q \parc q'$, 
$p \leftm p' \bisim q \leftm q'$, $p \commm p' \bisim q \commm q'$,  
$\encap{H}(p) \bisim \encap{H}(q)$, and 
$\state{s}(p) \bisim \state{s}(q)$.
\end{proposition}
\begin{proof}
The proof goes along the same line as the proof of 
Proposition~\ref{proposition-congr-ctBPAps}.
\qed
\end{proof}

\ctACPps\textup{{+}SO} is sound with respect to $\bisim$ for equations 
between closed terms.
\begin{theorem}[Soundness]
\label{theorem-soundness-ctACPps+SO}
For all closed \ctACPps\textup{{+}SO} terms $p,q$ of sort $\Proc$, 
\mbox{$p = q$} is derivable from the axioms of \ctACPps\textup{{+}SO} 
only if $p \bisim q$.
\end{theorem}
\begin{proof}
The proof goes along the same line as the proof of 
Theorem~\ref{theorem-soundness-ctACPps}.
\qed
\end{proof}

\ctACPps\textup{{+}SO} is complete with respect to $\bisim$ for 
equations between closed terms.
\begin{theorem}[Completeness]
\label{theorem-completeness-ctACPps+SO}
For all closed \ctACPps\textup{{+}SO} terms $p,q$ of sort $\Proc$, 
$p = q$ is derivable from the axioms of \ctACPps\textup{{+}SO} if 
$p \bisim q$.
\end{theorem}
\begin{proof}
The proof goes along the same line as the proof of 
Theorem~\ref{theorem-completeness-ctACPps}.
\qed
\end{proof}

\section{Guarded Recursion}
\label{sect-ctACPps+REC}

In order to allow for the description of processes without a finite 
upper bound to the number of actions that it can perform, we add in this 
section guarded recursion to \ctACPps\ and \ctACPps\textup{{+}SO}.
The resulting theories are called \ctACPps\textup{{+}REC} and 
\ctACPps\textup{{+}SO{+}REC}, respectively.

A \emph{recursive specification} over \ctACPps\ is a set of 
\emph{recursion} equations $E = \set{X = t_X \where X \in V}$ where $V$ 
is a set of variables of sort $\Proc$ and each $t_X$ is a term of sort 
$\Proc$ that only contains variables from $V$.
We write $\vars(E)$ for the set of all variables that occur on the
left-hand side of an equation in $E$.
A \emph{solution} of a recursive specification $E$ is a set of processes
(in some model of \ctACPps) $\set{P_X \where X \in \vars(E)}$ such that 
the equations of $E$ hold if, for all $X \in \vars(E)$, $X$ stands for 
$P_X$.

Let $t$ be a \ctACPps\ term of sort $\Proc$ containing a variable $X$.
We call an occurrence of $X$ in $t$ \emph{guarded} if $t$ has a subterm
of the form $a \seqc t'$, where $a \in \Act$, with $t'$ containing this 
occurrence of $X$.
A recursive specification $E$ over \ctACPps\ is called a \emph{guarded} 
recursive specification if all occurrences of variables in the 
right-hand sides of its equations are guarded or it can be rewritten to 
such a recursive specification using the axioms of \ctACPps\ in either 
direction and/or the equations in $E$ from left to right.
We are only interested in a model of \ctACPps\ in which guarded 
recursive specifications have unique solutions.

For each guarded recursive specification $E$ over \ctACPps\ and each 
variable $X \in \vars(E)$, we add a constant of sort $\Proc$, standing 
for the unique solution of $E$ for $X$, to the constants of \ctACPps.
This constant is denoted by $\rec{X}{E}$.

We will use the following notation.
Let $t$ be a \ctACPps\ term of sort $\Proc$ and $E$ be a guarded 
recursive specification over \ctACPps.
Then we write $\rec{t}{E}$ for $t$ with, for all $X \in \vars(E)$, all
occurrences of $X$ in $t$ replaced by $\rec{X}{E}$.

The additional axioms for guarded recursion are the equations given
in Table~\ref{axioms-REC}.%
\begin{table}[!tb]
\caption{Axioms for guarded recursion}
\label{axioms-REC}
\begin{eqntbl}
\begin{saxcol}
\rec{X}{E} = \rec{t_X}{E} & \mif X = t_X \in E  & \ax{RDP} \\
E \Limpl X  = \rec{X}{E}   & \mif X \in \vars(E) & \ax{RSP}
\end{saxcol}
\end{eqntbl}
\end{table}
In this table, $X$, $t_X$, and $E$ stand for an arbitrary variable
of sort $\Proc$, an arbitrary \ctACPps\ term, and an arbitrary guarded 
recursive specification over \ctACPps, respectively. 
Side conditions are added to restrict the variables, terms and guarded 
recursive specifications for which $X$, $t_X$ and $E$ stand.
The additional axioms for guarded recursion are known as the recursive
definition principle (RDP) and the recursive specification principle
(RSP).
The equations $\rec{X}{E} = \rec{t_X}{E}$ for a fixed $E$ express that
the constants $\rec{X}{E}$ make up a solution of $E$.
The conditional equations $E \Limpl X = \rec{X}{E}$ express that this
solution is the only one.

The additional transition rules for the constants $\rec{X}{E}$ are given 
in Table~\ref{sos-REC}.
\begin{table}[!tb]
\caption{Transition rules for guarded recursion}
\label{sos-REC}
\begin{ruletbl}
\RuleC
{\aterm{\rec{t_X}{E}}{\gact{\phi}{a}}}
{\aterm{\rec{X}{E}}{\gact{\phi}{a}}}
{X \!=\! t_X \,\in\, E}
\qquad
\RuleC
{\astep{\rec{t_X}{E}}{\gact{\phi}{a}}{x'}}
{\astep{\rec{X}{E}}{\gact{\phi}{a}}{x'}}
{X \!=\! t_X \,\in\, E}
\\
\RuleC
{\rsgn{\rec{t_X}{E}} = \phi}
{\rsgn{\rec{X}{E}} = \phi}
{X \!=\! t_X \,\in\, E}
\end{ruletbl}
\end{table}
In this table, $X$, $t_X$ and $E$ stand for an arbitrary variable
of sort $\Proc$, an arbitrary \ctACPps\ term and an arbitrary guarded 
recursive specification over \ctACPps, respectively. 

Bisimulation equivalence is a congruence with respect to  the operators 
of \ctACPps{+}REC.
\begin{proposition}[Congruence]
\label{proposition-congr-ctACPps+REC}
For all closed \ctACPps\textup{{+}REC} terms $p,q,p',q'$ of sort $\Proc$ 
and closed \ctACPps\textup{{+}REC} terms $\phi$ of sort $\Prop$, 
$p \bisim q$ and $p' \bisim q'$ imply $p \altc p' \bisim q \altc q'$, 
$p \seqc p' \bisim q \seqc q'$, $\phi \gc p \bisim \phi \gc q$, 
$\phi \emi p \bisim \phi \emi q$, $p \parc p' \bisim q \parc q'$, 
$p \leftm p' \bisim q \leftm q'$, $p \commm p' \bisim q \commm q'$,  
$\encap{H}(p) \bisim \encap{H}(q)$.
\end{proposition}
\begin{proof}
The proof goes along the same line as the proof of 
Proposition~\ref{proposition-congr-ctBPAps}.
\qed
\end{proof}

\ctACPps\textup{{+}REC} is sound with respect to $\bisim$ for equations 
between closed terms.
\begin{theorem}[Soundness]
\label{theorem-soundness-ctACPps+REC}
For all closed \ctACPps\textup{{+}REC} terms $p,q$ of sort $\Proc$, 
\mbox{$p = q$} is derivable from the axioms of \ctACPps\textup{{+}REC} 
only if $p \bisim q$.
\end{theorem}
\begin{proof}
Because of Proposition~\ref{proposition-congr-ctACPps+REC}, it is 
sufficient to prove the theorem for all closed \ctACPps{+}REC terms $p$ 
and $q$ for which $p = q$ is a closed substitution instance of an axiom 
of \ctACPps{+}REC.
With the exception of the closed substitution instances of RSP, the 
proof goes along the same line as the proof of 
Theorem~\ref{theorem-soundness-ctACPps}.
The proof of the validity of RSP is rather involved. 
We confine ourselves to a very brief outline of the proof.
The transition rules for \ctACPps{+}REC determines a transition system 
for each process that can be denoted by a closed \ctACPps{+}REC term of 
sort $\Proc$.
A model of \ctACPps{+}REC based on these transition systems can be 
constructed along the same line as the models of a generalization of 
ACPps constructed in~\cite{BM05a}.
An equation $p = q$ between closed \ctACPps{+}REC terms holds in this 
model iff $p \bisim q$.
Based on this model, the validity of RSP can be proved along the same 
line as in the proof of Theorem~10 from~\cite{BM05a}.
The underlying ideas of that proof originate largely from~\cite{BBK87b}.
\qed
\end{proof}

Guarded recursion can be added to \ctACPps{+}SO in the same way as it is
added to \ctACPps\ above, resulting in \ctACPps{+}SO{+}REC.
It is easy to see that the above results, i.e.\ 
Proposition~\ref{proposition-congr-ctACPps+REC} 
and Theorem~\ref{theorem-soundness-ctACPps+REC},
go through for \ctACPps{+}SO{+}REC.

Completeness of \ctACPps\textup{{+}REC} and \ctACPps\textup{{+}SO{+}REC} 
with respect to $\bisim$ for equations between closed terms can be 
obtained by restriction to the finite linear recursive specifications,
i.e.\ the guarded recursive specifications with finitely many recursion 
equations where the right-hand side of each recursion equation can be 
written in the form
$\chi \emi \dead \altc
 \vAltc{i \in \set{1,\ldots,n}} \phi_i \gc a_i \seqc X_i \altc
 \vAltc{j \in \set{1,\ldots,m}} \psi_j \gc b_j$,
where $n,m \in \Nat$, where $\chi \notin [\False]$,
where $\phi_i \notin [\False]$, $a_i \in \Act$, and $X_i$ is variable of
sort $\Proc$ for all $i \in \set{1,\ldots,n}$, and
where $\psi_j \notin [\False]$ and $b_j \in \Act$ for all 
$j \in \set{1,\ldots,m}$.

\section{Concluding Remarks}
\label{sect-concl}

We have presented \ctACPps, a version of ACPps built on a paraconsistent 
pro\-positional logic called LP$^{\IImpl,\False}$.
\ctACPps\ deals with processes with possibly self-contradictory states 
by means of this paraconsistent logic.
To our knowledge, processes with possibly self-contradictory states have 
not been dealt with in any theory or model of processes.
This leaves nothing to be said about related work.
However, it is worth mentioning that the need for a theory or model of 
processes with possibly self-contradictory states was already expressed 
in~\cite{Hew08a}.

In order to streamline the presentation of \ctACPps, we have left out 
the terminal signal emission operator, the global signal emission 
operator, and the root signal operator of ACPps and also the additional 
operators introduced in~\cite{BB94b} other than the state operators. 
To our knowledge, these are exactly the operators that have not been 
used in any work based on ACPps.
The root signal operator is an auxiliary operator which can be dispensed
with and the global signal emission operator is an auxiliary operator 
which can be dispensed with in the absence of the terminal signal 
emission operator.
The terminal signal emission operator makes it possible to express that
a proposition holds at the termination of a process.

\ctACPps\ is a contradiction-tolerant version of ACPps~\cite{BB94b}.
ACPps itself can be viewed as a simplification and specialization of 
ACPS~\cite{BB92c}.
The simplification consists of the use of conditions instead of special 
actions to observe signals.
The specialization consists of the use of the set of all propositions 
with propositional variables from a given set instead of an arbitrary
free Boolean algebra over a given set of generators.
Later, the generalization of ACPps to arbitrary such Boolean algebras 
has been treated in~\cite{BM05a}.
Moreover, a timed version of ACPps has been used in~\cite{BM03a} as the 
basis of a process algebra for hybrid systems and
a timed version of ACPps has been used in~\cite{BMU98a} to give a 
semantics to a specification language that was widely used in 
telecommunications at the time.

Timed versions of \ctACPps\ may be useful in various applications.
We believe that they can be obtained by combining \ctACPps\ with a timed 
version of ACP, such as ACP$^\mathrm{drt}$ or ACP$^\mathrm{srt}$ 
from~\cite{BM02a}, in much the same way as timed versions of ACPps have 
been obtained in~\cite{BM03a,BMU98a}.
Because idling of processes is taken into account, two forms of the 
guarded command operator can be distinguished in these timed versions,
namely a non-waiting form and a waiting form (see e.g.~\cite{BMU98a}).
A version of \ctACPps\ with abstraction features like in \ACP$^\tau$ 
(see e.g.~\cite{BW90}) may be useful in various applications as well.
Working out a timed version of \ctACPps\ and working out a version of 
\ctACPps\ with abstraction features are options for further work.
It is very important that case studies are carried out in conjunction
with the theoretical work just mentioned to assess the degree of 
usefulness in practical applications.

LP$^{\IImpl,\False}$ is Blok-Pigozzi algebraizable.
However, although there must exist one, a conditional-equa\-tional 
axiomatization of the algebras concerned has not yet been devised.
Owing to this, the equations derivable in \ctACPps\ cannot always be 
derived by equational reasoning only. 
Another option for further work is devising the axiomatization 
referred to.

\subsection*{Acknowledgements}
We thank two anonymous referees for carefully reading a preliminary 
version of this paper and for suggesting improvements of the 
presentation of the paper.

\bibliographystyle{splncs03}
\bibliography{PA,PCL}

\end{document}